\documentclass[english,twocolumn,10pt,superscriptaddress, prl]{revtex4-2}
\usepackage[T1]{fontenc}
\usepackage[latin9]{inputenc}
\usepackage{color}
\usepackage{babel}
\usepackage{amstext}
\usepackage{amssymb}
\usepackage{graphicx}
\usepackage{esint}
\usepackage[unicode=true,pdfusetitle,
 bookmarks=true,bookmarksnumbered=false,bookmarksopen=false,
 breaklinks=false,pdfborder={0 0 0},pdfborderstyle={},backref=false,colorlinks=true]
 {hyperref}
\hypersetup{
 linkcolor=blue, citecolor=cyan}

\makeatletter
\usepackage{amsthm}

\usepackage{xcolor}

\newtheorem{theorem}{Theorem}

\makeatother

\begin{document}
\title{Stability and causality criteria in linear mode analysis: \\Stability
means causality}
\author{Dong-Lin Wang}
\email{donglinwang@mail.ustc.edu.cn}

\affiliation{Department of Modern Physics, University of Science and Technology
of China, Anhui 230026, China}
\author{Shi Pu}
\email{shipu@ustc.edu.cn}

\affiliation{Department of Modern Physics, University of Science and Technology
of China, Anhui 230026, China}
\begin{abstract}
Causality and stability are fundamental requirements for the differential
equations describing predictable relativistic many-body systems. In
this work, we investigate the stability and causality criteria in
linear mode analysis. We discuss the updated stability criterion in
$3+1$-dimensional systems and introduce the improved sufficient criterion
for causality. Our findings clearly demonstrate that stability implies
causality in linear mode analysis. Furthermore, based on the theorems
present in this work, we conclude that if updated stability criterion
and improved causality criterion are fulfilled in one inertial frame
of reference (IFR), they hold for all IFR.
\end{abstract}
\maketitle

\paragraph*{Introduction ---}


In modern physics, the space-time evolution of predictable relativistic
many-body systems is typically described using differential equations.
These differential equations must adhere to the principle of causality
as required by the theory of relativity \citep{Wald:1984rg}. Causality
means that signals or information cannot propagate faster than the
speed of light. While the physical interpretation of causality is
well understood, the establishment of well-defined criteria for causality
in various relativistic systems is still rarely discussed. In the
context of quantum field theory, it translates into the requirement
that the commutators of local operators vanish outside the light cone
\citep{Gell-Mann:1954ttj}. For many macroscopic relativistic many-body
systems, it is still challenging to derive the general sufficient
and necessary criteria for causality, based on our current understanding
(also see Refs. \citep{Bemfica:2020zjp,Gavassino:2021kjm,Heller:2022ejw}
and the references therein, for recent developments for the causality
criteria in relativistic hydrodynamics).

Another essential prerequisite is stability, characterized by minor
perturbations in a state gradually diminishing over time. Typically,
one can find solutions to the differential equations near equilibrium
or eigenstates. Stability ensures that perturbations near equilibrium
or eigenstates return to their respective states. Noteworthy examples
for stability include circular orbits in the classical central force
problem \citep{landau1976:me}, particular solutions in general relativity
\citep{Abbott:1981ff}, equilibrium states in thermodynamics \citep{Landau:1980mil},
and quantum states in quantum mechanics \citep{Peres:1984pd,Bender:1998ke},
among others.


Analyzing stability and causality helps determine the physical viability
of a theory and provides additional constraints on its parameters.
One famous example is relativistic hydrodynamics. The conventional
relativistic hydrodynamics up to the first order in the gradient expansion
is found to be acausal and unstable \citep{Hiscock:1985zz}. It has,
therefore, been extended to second-order formalisms such as the M\"uller-Israel-Stewart
(MIS) theory \citep{Israel:1979}, Baier-Romatschke-Son-Starinets-Stephanov
theory \citep{Baier:2007ix}, and Denicol-Niemi-Molnar-Rischke theory
\citep{Denicol:2012cn}. In addition, a stable and causal generalized
first-order formalism, known as the Bemfica-Disconzi-Noronha-Kovtun
theory has also been established \citep{Bemfica:2017wps,Kovtun:2019hdm}.
Comprehensive discussions of the causality and stability conditions
for these theories can be found in Refs. \citep{Floerchinger:2017cii,Bemfica:2020zjp}
and the references therein. Moreover, recent interest has emerged
in the stability and causality of effective field theory for hydrodynamics
\citep{Mullins:2023tjg}.


Physically, causality and stability are intertwined. Causality also
implies that all physical observables must reside within the light
cone. Assume that there exists a stable mode in one IFR, in which
perturbation decays when $t-t_{0}>0$ with $t_{0}$ being initial
time. If the perturbation propagates out of the light cone, we will
observe that $t^{\prime}-t_{0}^{\prime}<0$ in another IFR due to
a Lorentz transformation. This means that the perturbation grows with
time in the second IFR. Therefore, it is concluded that acausality
leads to an unstable mode \citep{Gavassino:2021owo}. The above argument
reveals the connection between acausality and unstable modes; the
relationship between the causality and stability is worth being discussed
deeply. Although the causality and stability conditions overlap in
some theories \cite{Pu:2009fj}, the criteria for causality and stability
appear significantly different and seemingly independent. A natural
question arises as to how we can establish a connection between the
stability and causality criteria.


Because of the complexity involved, analyzing the causality and stability
of complete differential equations can be challenging. One common
approach is to employ linear mode analysis. Despite the significant
progress made in recent works \citep{Bemfica:2020zjp,Gavassino:2021owo}
regarding stability-causality relations, it remains unclear how to
deduce these remarkable findings directly from the basic stability
and causality criteria in linear mode analysis.

Let us start the story from a general overview of linear mode analysis.
Following that, we discuss the updated stability criteria and introduce
improved causality criteria. Finally, we uncover the connection between
stability and causality in linear mode analysis based on our findings.

\vspace{1em}

\paragraph*{Linear mode analysis and conventional causality and stability criteria
---}

Linear mode analysis is a widely used method for investigating the
stability and causality of differential equations. Here, we take the
relativistic hydrodynamics as an illustrative example to introduce
the basic idea for linear mode analysis. The main differential equations
for relativistic hydrodynamics are the energy-momentum and current
conservation equations, expanded in the gradient expansion. In the
linear mode analysis, one considers the perturbations of independent
macroscopic variables within the system, e.g. energy density $\delta e$,
number density $\delta\rho$, etc., near the equilibrium. Generally,
the hydrodynamic conservation equations for $n$ independent perturbations,
$\varphi(t,\vec{x})=(\delta e,\delta\rho,...)^{\mathrm{T}}$, on top
of the irrotational equilibrium state can be formulated as linear
partial differential equations \citep{Hiscock:1985zz}
\begin{equation}
\partial_{t}\varphi(t,\vec{x})+\mathrm{M}(-i\partial)\varphi(t,\vec{x})=0,\label{eq:StandardFormFirstOrder}
\end{equation}
where matrix $\mathrm{M}(-i\partial)$ is a polynomial of the space
derivative $\partial_{i}$, $\mathrm{M}(-i\partial)=\sum_{j=0}^{N}\mathrm{M}_{(j)}^{i_{1},i_{2},...,i_{j}}\partial_{i_{1}}\partial_{i_{2}}...\partial_{i_{j}}$
with $N\geq0$ being a finite integer and $\mathrm{M}_{(j)}^{i_{1},i_{2},...,i_{j}}$
being a constant $n\times n$ matrix. For simplicity, plane-wave type
perturbations are often adopted in the linear mode analysis, 
\begin{equation}
\varphi=\varphi_{0}e^{-i\omega t+i\vec{k}\cdot\vec{x}},\quad\varphi_{0}=\textrm{const}.
\end{equation}
Subsequently, the nonzero solutions to the differential equations
(\ref{eq:StandardFormFirstOrder}) exist if and only if 
\begin{equation}
0=\mathcal{P}(\omega,\vec{k})\equiv\det[\omega+i\mathrm{M}(\vec{k})].\label{eq:PolynomialP}
\end{equation}
The eigenvalues of $-i\mathrm{M}(\vec{k})$ yield the dispersion relations,
denoted as $\omega_{i}=\omega_{i}(\vec{k})$, $i=1,2,...,n$. Naturally,
the dispersion relations should be constrained by physical requirements:
\begin{itemize}
\item Stability: All perturbations $|\varphi(0,\vec{x})|$ cannot grow exponentially
with time. It gives the conventional stability criterion,  
\begin{equation}
\mathrm{Im}\ \omega\leq0,\qquad\textrm{for }\vec{k}\in\mathbb{R}^{3}.\label{eq:ineq_orginal_01}
\end{equation}
\item Causality: The influence of $\varphi(0,\vec{x})$ propagates no faster
than the speed of light. A widely accepted asymptotic causality criterion
is \citep{Krotscheck1978CausalityC}
\begin{equation}
\lim_{|\vec{k}|\rightarrow+\infty}\left\{ \frac{|\mathrm{Re}\ \omega|}{|\vec{k}|}\leq1,\textrm{ |\ensuremath{\omega}/\ensuremath{\vec{k}}| is bounded}\right\} ,\;\vec{k}\in\mathbb{R}^{3}.\label{eq:WrongC}
\end{equation}
\end{itemize}
Here, we define the vector norm $|V|\equiv\left(\sum_{i=1}^{j}V_{i}V_{i}^{*}\right)^{1/2}$
for any complex vector $V=(V_{1},V_{2},...,V_{j})$ and adopt notation
$\omega=\omega_{i}(\vec{k})$, $i=1,2,...,n$ for brevity moving forward.
The causality and stability criteria mentioned above are intuitive,
but they are not flawless.
\begin{itemize}
\item A practical challenge arises: the conventional causality and stability
criteria (\ref{eq:ineq_orginal_01}) and (\ref{eq:WrongC}) depend
on IFR. Commonly, the causality and stability conditions are first
derived from the criteria (\ref{eq:ineq_orginal_01}) and (\ref{eq:WrongC})
in the rest frame. Then, the verification of these criteria in other
IFR follows. However, this process of examining conditions across
different frames is frequently burdensome.
\item A concern arises: the conventional causality criterion (\ref{eq:WrongC})
proves to be inadequate in guaranteeing causality \citep{Gavassino:2023mad}.
\item A question arises: does stability imply causality? Furthermore, what
constitutes the relationship between the stability and the causality
criteria?
\end{itemize}
The aim of this work is to provide improved causality criteria, simplify
the steps for Lorentz transformation, bridge the gap between stability
and causality criteria, and reveal the profound stability-causality
relations. We propose an improved sufficient criterion for causality.
This, combined with new insights into the covariance of the stability
and causality criteria, allows for the immediate derivation of significant
stability-causality relations in linear mode analysis. Before further
discussion, we emphasize that in the current study we concentrate
on the systems with well-defined inertial frames. When the system
is far from equilibrium or incorporates certain quantum effects \cite{Arnold:2014jva},
the local rest frames may become ill-defined and the following criteria
may be inapplicable.


\vspace{1em}

\paragraph*{Updated stability criterion for a $3+1$-dimensional relativistic
system ---}

The updated stability criterion for a $3+1$-dimensional relativistic
system is
\begin{equation}
\mathrm{Im}\ \omega\leq|\mathrm{Im}\ \vec{k}|,\quad\text{for }\vec{k}\in\mathbb{C}^{3}.\label{eq:NewBound}
\end{equation}
The inequality (\ref{eq:NewBound}) is introduced within a $1+1$
dimensional system by imposing the causality on the retarded two-point
function in stable systems, as proposed in Ref. \citep{Heller:2022ejw},
and is subsequently proved as the necessary condition for stability
across all IFR \citep{Gavassino:2023myj}.

The inequality (\ref{eq:NewBound}) can be proven by employing a contradiction
approach. Assume that $\mathrm{Im}\ \omega>|\mathrm{Im}\ \vec{k}|$
holds true for a specific $\vec{k}$ in one IFR, denoted as $K$,
even while maintaining system stability. We perform a special Lorentz
transformation from frame $K$ to another IFR $K^{\prime}$, characterized
by a velocity $\vec{v}=(\mathrm{Im}\ \vec{k})/(\mathrm{Im\ }\omega)$
relative to $K$. In frame $K^{\prime}$, as $(\omega^{\prime},\vec{k}^{\prime})^{\mathrm{T}}=\Lambda(\vec{v})\cdot(\omega,\vec{k})^{\mathrm{T}}$
with $\Lambda(\vec{v})$ being the Lorentz transformation matrix,
we observe that $\mathrm{Im\ }\vec{k}^{\prime}=0$ and $\mathrm{Im}\ \omega^{\prime}=\gamma^{-1}\mathrm{Im}\ \omega>0$
with $\gamma$ being Lorentz factor. This implies the existence of
an unstable mode, which violates the stability requirement (\ref{eq:ineq_orginal_01}),
within the frame $K^{\prime}$. This indicates that the assumption
$\mathrm{Im}\ \omega>|\mathrm{Im}\ \vec{k}|$, made in any IFR, can
render the system unstable.

Furthermore, the subsequent theorem significantly streamlines the
intricate calculations associated with transformations between distinct
IFR in linear mode analysis.

\begin{theorem} \label{TheoremS1}

The stability criterion (\ref{eq:NewBound}) holds true across all
IFR if it is satisfied in a single IFR.

\end{theorem}


\begin{proof}

It can also be proven by employing a contradiction approach. Assume
that  $\mathrm{Im}\ \omega^{\prime}\leq|\mathrm{Im}\ \vec{k}^{\prime}|$
holds in a IFR $K^{\prime}$. Let us suppose that this inequality
is violated in another IFR $K^{\prime\prime}$, i.e., there exists
$\vec{k}^{\prime\prime}\in\mathbb{C}^{3}$ within frame $K^{\prime\prime}$
for which $\mathrm{Im}\ \omega^{\prime\prime}>|\mathrm{Im}\ \vec{k}^{\prime\prime}|$.
By the Lorentz transformation, $(\omega^{\prime\prime},\vec{k}^{\prime\prime})^{\mathrm{T}}=\Lambda(\vec{v})\cdot(\omega^{\prime},\vec{k}^{\prime})^{\mathrm{T}}$
where $\vec{v}$ represents the velocity of frame $K^{\prime\prime}$
relative to $K^{\prime}$, we find  
\begin{equation}
|\mathrm{Im}\ \omega^{\prime}|^{2}-|\mathrm{Im}\ \vec{k}^{\prime}|^{2}=|\mathrm{Im}\ \omega^{\prime\prime}|^{2}-|\mathrm{Im}\ \vec{k}^{\prime\prime}|^{2}>0,\label{eq:InvariantEq}
\end{equation}
and $\mathrm{Im}\ \omega^{\prime}=\gamma(\mathrm{Im}\ \omega^{\prime\prime}-\vec{v}\cdot\mathrm{Im}\ \vec{k}^{\prime\prime})>0$.
Thus, we arrive at $\mathrm{Im}\ \omega^{\prime}>|\mathrm{Im}\ \vec{k}^{\prime}|$,
which contradicts the original assumption. It means that frame $K^{\prime\prime}$
does not exist. Therefore, the inequality (\ref{eq:NewBound}) holds
across all IFR. \end{proof}

Theorem \ref{TheoremS1} can also be intuitively comprehended through
a geometric lens. Specifically, the vector $\mathrm{Im}(\omega,\vec{k})$
does not lay inside the future light cone, as discussed in Ref. \citep{Gavassino:2023myj}.


\vspace{1em}

\paragraph*{Improved sufficient criterion for causality ---}

Drawing inspiration from the stability criterion (\ref{eq:NewBound}),
we provide a new sufficient criterion for causality based on the theorem
below,

\begin{theorem} \label{TheoremSuffCausa}

Suppose that the initial data $\varphi(0,\vec{x})$ for differential
equations (\ref{eq:StandardFormFirstOrder}) is smooth with respect
to $\vec{x}$, and the volume of the support of $\varphi(0,\vec{x})$
is both finite and nonvanishing. If two constants $R>0$ and $b\in\mathbb{R}$
exist such that 
\begin{equation}
\mathrm{Im}\ \omega\leq|\mathrm{Im}\ \vec{k}|+b,\ \textrm{for }|\vec{k}|>R,\label{eq:SufCausality}
\end{equation}
then the influence of the initial data propagates with subluminal
speed.

\end{theorem}

Before proving the aforementioned theorem, we intend to present a
simplified equivalent version of criterion (\ref{eq:SufCausality}),
denoted as follows:\emph{ If condition (\ref{eq:SufCausality}) is
fulfilled, then there exists an additional real constant $b^{\prime}\geq b$
such that }
\begin{equation}
\mathrm{Im}\ \omega\leq|\mathrm{Im}\ \vec{k}|+b^{\prime},\ \textrm{for}\ \vec{k}\in\mathbb{C}^{3}.\label{eq:OmegaBound}
\end{equation}
Let us deduce the inequality (\ref{eq:OmegaBound}) from inequality
(\ref{eq:SufCausality}). We notice that the $\omega=\omega(\vec{k})$
must be finite for any finite $|\vec{k}|$. This is because the dispersion
relation from (\ref{eq:PolynomialP}), represented by $\mathcal{P}(\omega,\vec{k})=\omega^{n}+\sum_{m=0}^{n-1}a_{m}(\vec{k})\omega^{m}=0$
with $a_{m}(\vec{k})$ being a polynomial of $\vec{k}$, will not
yield an infinite $\omega$ for finite $\vec{k}$. Consequently, there
exists a sufficiently large positive constant $b^{\prime}$ such that
$b^{\prime}\geq b$ and $\mathrm{Im}\ \omega\leq|\mathrm{Im}\ \vec{k}|+b^{\prime}$
for $|\vec{k}|\leq R$. Here we have implicitly assumed that the
perturbed equations can be written as the form (\ref{eq:StandardFormFirstOrder}),
which has already ruled out some acausal equations, e.g., the Benjamin-Bona-Mahony
equation \citep{Gavassino:2023mad}. Now, let us direct our attention
toward proving the causality criterion (\ref{eq:SufCausality}) with
the help of its simplified equivalent version (\ref{eq:OmegaBound}).

\begin{proof}

Suppose that the initial data $\varphi(0,\vec{x})$, possessing finite
volume, are enclosed within the closed ball centered at $\vec{x}=0$
and having a radius of $L>0$.  By employing the general solution
of (\ref{eq:StandardFormFirstOrder}), we can express $\varphi(t,\vec{x})$
with $t>0$ as follows:
\begin{equation}
\varphi(t,\vec{x})=\int_{\mathbb{R}^{3}}\frac{d^{3}k}{(2\pi)^{3}}e^{i\vec{k}\cdot\vec{x}}e^{-\mathrm{M}(\vec{k})t}\tilde{\varphi}(0,\vec{k}),\label{eq:Phitx}
\end{equation}
where $\tilde{\varphi}(t,\vec{k})\equiv\int d^{3}xe^{-i\vec{k}\cdot\vec{x}}\varphi(t,\vec{x})$
represents the Fourier transformation of $\varphi(t,\vec{x})$. 
In this case, the causality means the perturbation $\varphi(t,\vec{x})$
cannot persist beyond the region $|\vec{x}|>L+ct$ with any finite
$t>0$, where $c=1$ is the speed of light \citep{Wald:1984rg}. Hence,
our task is to demonstrate that $\varphi(t,\vec{x})=0$ within the
region $|\vec{x}|>L+t$ with any finite $t>0$, given that the dispersion
relations adhere to the inequality (\ref{eq:SufCausality}).

The key lies in employing the following two inequalities:
\begin{eqnarray}
||e^{-\mathrm{M}(\vec{k})t}|| & \leq & \frac{a_{1}}{\epsilon^{n-1}}(1+|\vec{k}|^{N(n-1)})e^{[\lambda(\vec{k})+\epsilon]t},\label{eq:Estimate1}\\
|\tilde{\varphi}(0,\vec{k})| & \leq & \frac{a_{2}}{1+|\vec{k}|^{N(n-1)+4}}e^{L|\mathrm{Im}\ \vec{k}|},\label{eq:Estimate2}
\end{eqnarray}
where $\lambda(\vec{k})\equiv\max_{i}\{\mathrm{Im}\ \omega_{i}(\vec{k})\}$,
$\epsilon\in(0,1)$, and $a_{1},a_{2}$ are independent of $t$, $\epsilon$,
and $\vec{k}$. Here, the norm $||\cdot||$ is the spectral norm of
matrix \citep{Datta2010:NLinear}. The inequality (\ref{eq:Estimate2})
can be obtained by performing integration by parts $N(n-1)+4$ times
in $\tilde{\varphi}(t,\vec{k})=\int d^{3}xe^{-i\vec{k}\cdot\vec{x}}\varphi(t,\vec{x})$
\citep{Lax:2006Hyperbolic}.  Proof of inequality (\ref{eq:Estimate1})
utilizes the Cauchy integral formula for matrices \citep{Evans2010PartialDE},
and the details can be found in the Supplemental Material. 

\begin{figure}[t]
\noindent \begin{centering}
\includegraphics[scale=0.4]{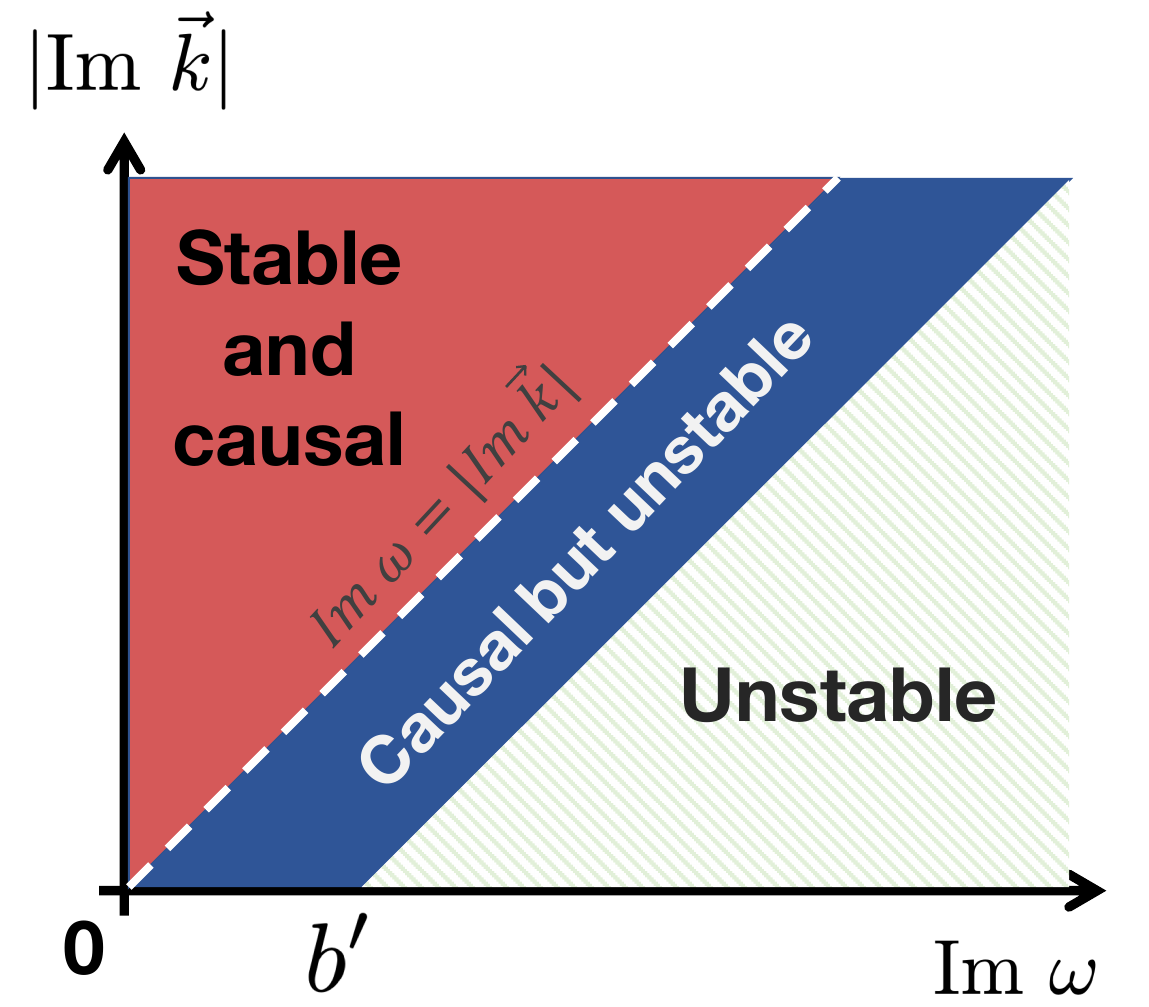}
\par\end{centering}
\caption{Illustration for stability and causality criteria. \label{fig:Illustration-for-stability}}
\end{figure}

Assisted by these two inequalities, we estimate the integral (\ref{eq:Phitx})
for $|\vec{x}|>L+t$, employing a method akin to the proof outlined
in Theorem 3.1 of Ref. \citep{Lax:2006Hyperbolic}. Upon the fulfillment
of the inequality (\ref{eq:SufCausality}) or its simplified equivalent
version (\ref{eq:OmegaBound}), we have $\lambda(\vec{k})\leq|\mathrm{Im}\ \vec{k}|+b^{\prime}$.
Given that the integrand in (\ref{eq:Phitx}) is an entire analytic
function, we change the path $\vec{k}\rightarrow\vec{k}+ir\vec{x}$
and apply Cauchy's theorem, 
\begin{eqnarray}
\left|\varphi(t,\vec{x})\right| & = & \left|\int_{\mathbb{R}^{3}}d^{3}ke^{i(\vec{k}+ir\vec{x})\cdot\vec{x}}e^{-\mathrm{M}(\vec{k}+ir\vec{x})t}\tilde{\varphi}(0,\vec{k}+ir\vec{x})\right|\nonumber \\
 & \propto & e^{-r|\vec{x}|(|\vec{x}|-t-L)+(b^{\prime}+\epsilon)t}\rightarrow0,\label{eq:prove_02}
\end{eqnarray}
as $r\rightarrow+\infty$, where we have used (\ref{eq:Estimate1}),
(\ref{eq:Estimate2}), and the inequality $|e^{-\mathrm{M}(\vec{k})t}\tilde{\varphi}(0,\vec{k})|\leq||e^{-\mathrm{M}(\vec{k})t}||\cdot|\tilde{\varphi}(0,\vec{k})|$
\citep{Datta2010:NLinear}. Hence, for any given finite $t>0$, we
deduce $\varphi(t,\vec{x})=0$ for $|\vec{x}|>L+t$ in accordance
with (\ref{eq:prove_02}), thereby affirming that the perturbation
signal propagates at a speed no greater than that of light. \end{proof}

Analogous to the stability criterion, the subsequent theorem facilitates
the extension of the causality criterion (\ref{eq:SufCausality})
or (\ref{eq:OmegaBound}) from one IFR to all IFR.

\begin{theorem} \label{TheoremCausOtherF}

The causality criterion (\ref{eq:SufCausality}) or (\ref{eq:OmegaBound})
holds true across all IFR if it is fulfilled in a single IFR.

\end{theorem}


\begin{proof}

We prove it by employing a contradiction approach again. Let us focus
on the inequality (\ref{eq:OmegaBound}). Suppose, by contradiction,
that (\ref{eq:OmegaBound}) holds in a IFR $K$, but is violated in
another IFR $K^{\prime}$, where $K^{\prime}$ moves with a velocity
$\vec{v}$ relative to $K$. Thus for any positive real constant $b^{\prime}>0$,
there exists a $\vec{k}^{\prime}\in\mathbb{C}^{3}$ such that $\mathrm{Im}\ \omega^{\prime}>|\mathrm{Im}\ \vec{k}^{\prime}|+b^{\prime}$
in frame $K^{\prime}$. Taking Lorentz transformation, we obtain $\mathrm{Im}\ \omega>|\mathrm{Im}\ \vec{k}|+\gamma b^{\prime}(1-|\vec{v}|)$,
which contradicts (\ref{eq:OmegaBound}) within frame $K$.  Therefore,
the frame $K^{\prime}$ does not exist. This completes the proof.
\end{proof}

In practice, one can verify the causality condition by taking $R\rightarrow\infty$
in inequality (\ref{eq:SufCausality}). In large $|\vec{k}|$ limit,
the typical dispersion relations exhibit behavior such as $\omega\propto\vec{v}_{0}\cdot\vec{k}+\mathcal{O}(|\vec{k}|^{0}),$
or $\omega\propto c_{0}(\vec{k}\cdot\vec{k})^{1/2}+\mathcal{O}(|\vec{k}|^{0})$,
where $c_{0}$ and $\vec{v}_{0}$ are real constants satisfying $|c_{0}|,|\vec{v}_{0}|\in[0,1]$.
 In this scenario, the sufficient criterion (\ref{eq:SufCausality})
holds, ensuring that the system maintains causality.

The\textcolor{red}{{} }new improved sufficient causality criterion (\textcolor{red}{\ref{eq:SufCausality}})
or (\ref{eq:OmegaBound}) is more stringent and preferable compared
to the conventional insufficient inequality (\textcolor{red}{\ref{eq:WrongC}}).
We show that the conventional criterion (\ref{eq:WrongC}) is automatically
fulfilled when the dispersion relations obey the inequality (\ref{eq:SufCausality})
in the Supplemental Material. However, the inequality (\ref{eq:SufCausality})
cannot be derived from (\ref{eq:WrongC}). For example, consider the
dispersion relation $\omega=k(1+i)/2$, which indeed satisfies the
inequality (\ref{eq:WrongC}), but does not obey the condition (\ref{eq:SufCausality}).
This specific case has been demonstrated to be acausal according to
Theorem 2.7 in Ref. \citep{Lax:2006Hyperbolic}. \textcolor{red}{}

The updated stability criterion (\ref{eq:NewBound}) and improved
causality criterion (\textcolor{red}{\ref{eq:SufCausality}}) or (\ref{eq:OmegaBound})
offer a straightforward way to reveal the profound relations between
stability and causality in linear mode analysis. Consequently, we
derive two corresponding conclusions.

\vspace{1em}

\paragraph*{Stability in all IFR means causality ---}

To illustrate the relationship between the stability and causality,
we depict the regions that satisfy the stability criterion (\ref{eq:NewBound})
and causality criterion (\ref{eq:OmegaBound}) in Fig. \ref{fig:Illustration-for-stability}.
Evidently, the causality criterion (\ref{eq:OmegaBound}) is inherently
satisfied when the stability criterion (\ref{eq:NewBound}) holds
true in all IFR. However, the reverse does not necessarily hold, e.g.
as presented by the blue region in Fig. \ref{fig:Illustration-for-stability}.
An example is the tachyon field equation \citep{Aharonov:1969vu},
whose dispersion relations fall inside the blue region (also see Supplemental
Material).\textcolor{blue}{{} }Therefore, stability across all IFR implies
causality, while causality does not necessarily entail stability.
To prevent any misinterpretation, it is crucial to highlight that
if the system has stability solely in specific IFR while being unstable
in others, the presence of stability alone does not guarantee causality.
The same conclusion has also been substantiated in Refs. \citep{Bemfica:2020zjp,Gavassino:2021owo}
through different approaches (also see a preliminary discussion in
Ref. \citep{Pu:2009fj}).

This finding is compatible with another significant observation in
Ref. \citep{Gavassino:2021kjm} referred to as ``thermodynamic stability
implies causality''. The conditions for thermodynamic stability delineated
as (i)-(iii) in Ref. \citep{Gavassino:2021kjm} maintain Lorentz invariance.
Clearly, if the thermodynamic stability holds true in a particular
IFR, it holds in any IFR. \textcolor{blue}{}


\vspace{0.5em}

\paragraph*{Stability and causality in one IFR implies stability and causality
across all IFR ---}

The subsequent theorem concerning stability-causality assessment across
various frames assists in mitigating the challenge in linear mode
analysis, where one is required to examine stability and causality
criteria in multiple frames. Similar conclusions have also been reported
in Refs. \citep{Bemfica:2020zjp,Gavassino:2021owo}, relying on the
premise of strong hyperbolicity or causality in all IFR.

\begin{theorem} \label{TheoremSCinMovingF}

If the system described by differential equation (\ref{eq:StandardFormFirstOrder})
exhibits stability and causality in one IFR, then it is stable and
causal in all IFR.

\end{theorem}

\begin{proof}

If differential equation (\ref{eq:StandardFormFirstOrder}) is stable
and causal in a particular IFR, then it follows that $\mathrm{Im}\ \omega\leq|\mathrm{Im}\ \vec{k}|$
within this frame (see Theorem 2 in Ref. \citep{Gavassino:2023myj}
for the $1+1$-dimensional case and the Supplemental Material for
the $3+1$-dimensional case). Consequently, combining Theorem \ref{TheoremS1}-\ref{TheoremCausOtherF}
presented in this work completes the proof. \end{proof}

Theorem \ref{TheoremSCinMovingF} provides a practical way for evaluating
stability and causality across all IFR. Initially, one selects a suitable
IFR, e.g. the rest frame in hydrodynamics. One can derive the stability
condition by analyzing conventional stability criterion (\ref{eq:ineq_orginal_01})
with the Routh-Hurwitz criterion \citep{Kovtun:2019hdm}, which is
often more straightforward than directly assessing stability criterion
(\ref{eq:NewBound}). Subsequently, the causality condition can be
verified using the causality criterion (\ref{eq:SufCausality}) as
$|\vec{k}|\rightarrow\infty$. To illustrate this point, we provide
an example concerning the stability and causality of MIS theory with
bulk viscous pressure only, presented in the Supplemental Material.
Interestingly, the calculations become straightforward in isotropic
systems by applying the theorem discussed in the coming paragraph.

\vspace{0.5em}

\paragraph*{Application: Asymptotic criterion for an isotropic system ---}

Given that numerous discussions revolve around causality and stability
within isotropic systems, such as the case of conventional relativistic
hydrodynamics in the rest frame, it is fitting for us to examine the
stability and causality conditions tailored for isotropic systems,
serving as a practical application. In an isotropic system, a simple
asymptotic criterion, as presented in the following theorem, becomes
a necessary condition for stability and a sufficient condition for
causality across all IFR.

\begin{theorem} \label{TheoremAsymp}

Considering $\vec{k}=k\hat{\mathrm{n}}$ where $k\in\mathbb{C}$,
$\hat{\mathrm{n}}\in\mathbb{C}^{3}$ and $\hat{\mathrm{n}}\cdot\hat{\mathrm{n}}=1$.
If the nonzero dispersion relations obtained from the equation (\ref{eq:PolynomialP})
are $\hat{\mathrm{n}}$-independent and satisfy inequality (\ref{eq:NewBound}),
then there exist only three asymptotic behaviors at $k\rightarrow\infty$,
\begin{eqnarray}
\omega & = & c_{1}k+d_{1}+\mathcal{O}(|k|^{a}),\label{eq:Iso_1}\\
\omega & = & c_{2}k^{-2m-1}+d_{2}k^{-2m-2}+\mathcal{O}(|k|^{a-2m-2}),\label{eq:Iso_2}\\
\omega & = & c_{3}k^{-2m}+\mathcal{O}(|k|^{a-2m}),\label{eq:Iso_3}
\end{eqnarray}
where $a<0$, $m=0,1,2,...$, and $c_{i},d_{i}$ are constants obeying
$c_{1}\in[-1,1]$, $\mathrm{Im}\ c_{2}=0$, $\mathrm{Im}\ c_{3}\leq0$,
$\mathrm{Im}\ d_{1,2}\leq0$.

\end{theorem}

The proof of the above theorem follows the asymptotic analysis in
Ref. \citep{bender1999advanced} and in the Supplemental Material.
We emphasize that the three dispersion relations (\ref{eq:Iso_1})-(\ref{eq:Iso_3})
mentioned above are necessary conditions for stability. If the asymptotic
behaviors of an isotropic system do not adhere to these conditions,
it might be causal but must be unstable. Interestingly, the three
dispersion relations (\ref{eq:Iso_1})-(\ref{eq:Iso_3}) satisfy the
conventional causality criterion (\ref{eq:WrongC}). This observation
helps explain why the conventional causality criterion (\ref{eq:WrongC})
has been considered a necessary condition for a covariantly stable
and causal isotropic system for a long time.

\vspace{1em}

\paragraph*{Summary ---}

In this work we have investigated the updated stability criterion
and improved causality criterion for the $3+1$-dimensional relativistic
system across all IFR. Notably, our findings indicate that the previously
widely used causality criterion (\ref{eq:WrongC}) needs to be substituted
with the improved asymptotic criterion (\ref{eq:SufCausality}) or
(\ref{eq:OmegaBound}). Based on Theorem \ref{TheoremS1}-\ref{TheoremCausOtherF},
we reveal the underlying connection between stability and causality
in linear mode analysis. Stability in all IFR implies causality, while
causality alone does not necessarily require stability. Furthermore,
if a system is stable and causal in one IFR, stability and causality
holds in all IFR. The findings alleviate the challenge of linear mode
analysis, which involves verifying the stability and causality conditions
across various frames. As an application, we also study the linear
stability and causality of the $3+1$-dimensional isotropic systems
and derive the new criterion (\ref{eq:Iso_1})-(\ref{eq:Iso_3}) that
are necessary for stability and sufficient for causality in all IFR.
Finally, it is important to emphasize that our theorems are model-independent
and can be applied to other relativistic systems beyond relativistic
hydrodynamics.


\vspace{0.5em}

\begin{acknowledgments}
We thank for J. Noronha, L. Gavassino, E. Grossi and P. Kovtun for
helpful discussions. This work is supported in part by the National
Key Research and Development Program of China under Contract No. 2022YFA1605500,
by the Chinese Academy of Sciences (CAS) under Grants No. YSBR-088
and by National Nature Science Foundation of China (NSFC) under Grants
No. 12075235 and No. 12135011. While this work was being completed,
we were informed of Ref. \citep{Hoult:2023clg} which works on a similar
topic and appeared on arXiv on the same day.
\end{acknowledgments}

\bibliographystyle{h-physrev}

\onecolumngrid
\newpage
\setcounter{equation}{0}
\setcounter{theorem}{0}
\renewcommand{\theequation}{S\arabic{equation}}
\begin{center}
\begin{center}
\textbf{\large Stability and causality criteria in linear mode analysis: stability means causality  \\ Supplemental Material }\\[.2cm]
 Dong-Lin Wang,$^{1}$ and Shi Pu $^{1}$ \\[.1cm]
{\itshape ${}^1$Department of Modern Physics, University of Science and Technology of China, Anhui 230026, China}
\end{center}
\setcounter{page}{1} 
\par\end{center}

\section{Causality, stability and related topics studied in other methods}

\subsection{Stability and two point Green functions \label{sec:ProofStabilityBound}}

We now follow Ref. \citep{Heller:2022ejw} to demonstrate that (\ref{eq:NewBound})
holds for any complex vector $\vec{k}$. Consider a retarded two point
Green function $G^{R}(x,y)$ for certain operators. We take $G^{R}(0,y)$
as a tempered distribution \citep{Heller:2022ejw}.  The causality
requires $G^{R}(0,y)$ is only nonzero in the past closed light cone
containing the boundary, i.e. $y^{0}\leq-|\vec{y}|<0$. The Fourier
transform of $G^{R}(0,y)$ is given by,
\begin{eqnarray}
\tilde{G}^{R}(\omega,\vec{k}) & \equiv & \int_{y^{0}\leq-|\vec{y}|<0}d^{4}yG^{R}(0,y)\times\exp[(\mathrm{Im}\ \omega)y^{0}\nonumber \\
 &  & -(\mathrm{Im}\ \vec{k})\cdot\vec{y}-i(\mathrm{Re}\ \omega)y^{0}+i(\mathrm{Re}\ \vec{k})\cdot\vec{y}].\label{eq:G_R}
\end{eqnarray}
Let us search for the analytic region for $\tilde{G}^{R}(\omega,\vec{k})$.
It means that the exponent in the right handed side of (\ref{eq:G_R})
must be suppressed. We find that 
\begin{eqnarray}
(\mathrm{Im}\ \omega)y^{0}-(\mathrm{Im}\ \vec{k})\cdot\vec{y} & \leq & (-\mathrm{Im}\ \omega+|\mathrm{Im}\ \vec{k}|)\cdot|\vec{y}|,\ \textrm{if }\mathrm{Im}\ \omega\geq0.
\end{eqnarray}
Once $\mathrm{Im}\ \omega>|\mathrm{Im}\ \vec{k}|$, it is obvious
that exponent is suppressed and $\tilde{G}^{R}(\omega,\vec{k})$ is
analytic. The dispersion relation $\omega=\omega(\vec{k})$ as the
non-analytic pole of $\tilde{G}^{R}(\omega,\vec{k})$ must be out
of region $\mathrm{Im}\ \omega>|\mathrm{Im}\ \vec{k}|$, i.e. it should
satisfy inequality (\ref{eq:NewBound}).

From the above discussion, we start from the requirement of causality
and get the stability condition. It implies that the stability and
causality are related.

\subsection{Extending theorem 2 in Ref. {[}24{]} to $3+1$ dimensional systems
\label{sec:3p1version}}

For completeness, here we write down the $3+1$ dimensional version
of Theorem 2 in Ref. \citep{Gavassino:2023myj}. The proof below is
very similar to the original proof for the $1+1$ dimensional propagation
in Ref. \citep{Gavassino:2023myj}.

\begin{theorem}

If a system described by the differential equations (\ref{eq:StandardFormFirstOrder})
is stable and causal in one IFR, then 
\[
\mathrm{Im}\ \omega\leq|\mathrm{Im}\ \vec{k}|,\quad\textrm{for }\vec{k}\in\mathbb{C}^{3},
\]
holds in this reference frame.

\end{theorem}

\begin{proof}

For any complex $\vec{k}$, the function $\varphi(t,\vec{x})=\epsilon e^{i(\vec{k}\cdot\vec{x}-\omega t)}$
with a small constant $\epsilon>0$ is the solution to (\ref{eq:StandardFormFirstOrder}).
When $\mathrm{Im}\ \vec{k}=0$, the norm $|\varphi(0,\vec{x})|$ is
small for any $\vec{x}$. The stability requires that $\varphi(t,\vec{x})$
cannot grow exponentially, then $\mathrm{Im}\ \omega\leq0=|\mathrm{Im}\ \vec{k}|$.

When $\mathrm{Im}\ \vec{k}\neq0$, we can choose the coordinates such
that $\mathrm{Im}\ \vec{k}=(|\mathrm{Im}\ \vec{k}|,0,0)$. We find
$|\varphi(0,\vec{x})|$ is unbounded as $x\rightarrow-\infty$, so
$\varphi(0,\vec{x})$ cannot be viewed as a small perturbation. To
avoid this problem, we define \citep{Gavassino:2021owo,Gavassino:2023myj}
\begin{equation}
\psi(0,\vec{x})\equiv\Theta(x)\varphi(0,\vec{x}),
\end{equation}
where $\Theta(x)\geq0$ is a smooth function that takes the value
of $0$ when $x\leq-1$ and the value of $1$ when $x\geq0$, with
a smooth transition in the interval $[-1,0]$. Notice that the vector
norm $|\psi(0,\vec{x})|$ is always small, which means $\psi(0,\vec{x})$
is a small perturbation. Thus we can impose the stability condition
on the solution $\psi(t,\vec{x})$, i.e., $\psi(t,\vec{x})$ cannot
grow exponentially for $t>0$. On the other hand, since $\psi(0,\vec{x})=\varphi(0,\vec{x})$
for $x\geq0$, the causality requires that \citep{Gavassino:2023myj}
\begin{equation}
|\psi(t,\vec{x})|_{x=t+\epsilon>t}=|\varphi(t,\vec{x})|_{x=t+\epsilon>t}\propto e^{-(|\mathrm{Im}\ \vec{k}|-\mathrm{Im}\ \omega)t}.
\end{equation}
Then the stability requirement gives $\mathrm{Im}\ \omega\leq|\mathrm{Im}\ \vec{k}|$.
\end{proof}

\subsection{Discussion on the conventional causality criterion \label{subsec:Relation-between-the}}

Here we will show that the conventional causality criterion (\ref{eq:WrongC})
can be deduced from the improved causality criterion (\ref{eq:SufCausality}).
The following theorem is the key to prove it.

\begin{theorem}

Let $\vec{k}=k\hat{\mathrm{n}}$ with $k\in\mathbb{C}$, $\hat{\mathrm{n}}\in\mathbb{R}^{3}$
and $\hat{\mathrm{n}}\cdot\hat{\mathrm{n}}=1$. If the nonzero dispersion
relations derived from the equation (\ref{eq:PolynomialP}) satisfy
inequality (\ref{eq:SufCausality}), then we have the asymptotic behavior
at $k\rightarrow\infty$: 
\begin{eqnarray}
\omega & = & C(\hat{\mathrm{n}})k+\mathcal{O}(|k|^{0}),\label{eq:CausalAsymB_1}
\end{eqnarray}
where $C(\hat{\mathrm{n}})\in[-1,1]$.

\end{theorem}

Obviously, the asymptotic behavior (\ref{eq:CausalAsymB_1}) obey
the conventional causality criterion (\ref{eq:WrongC}). Thus this
theorem implies that (\ref{eq:WrongC}) will automatically be satisfied
when the improved causality criterion (\ref{eq:SufCausality}) is
satisfied. In the following we give the detail of the proof of this
theorem.

\begin{proof}

By dominant balance \citep{bender1999advanced}, the nonzero asymptotic
solutions to the equation (\ref{eq:PolynomialP}) have the form 
\begin{equation}
\omega=C(\hat{\mathrm{n}})k^{\alpha}+\mathcal{O}(|k|^{\beta}),\quad\text{as }k\rightarrow\infty,
\end{equation}
where $\beta<\alpha$ and $C(\hat{\mathrm{n}})\neq0$ are independent
of $k$.

First, we prove that $\alpha\leq1$. From the inequality (\ref{eq:SufCausality})
we get 
\begin{eqnarray}
\mathrm{Im}[C(\hat{\mathrm{n}})k^{\alpha}+\mathcal{O}(|k|^{\beta})] & \leq & |\mathrm{Im}\ k|+b,\quad\textrm{for }|k|>R.
\end{eqnarray}
For any small $\epsilon>0$, there exists a large positive constant
$R_{1}$ such that $R_{1}\geq R$ and $\mathrm{Im}[\mathcal{O}(|k|^{\beta})]>-\epsilon|k|^{\alpha}$
for $|k|>R_{1}$. Then we have 
\begin{equation}
|k|^{\alpha}[|C(\hat{\mathrm{n}})|\sin(\alpha\theta_{k}+\theta_{c})-\epsilon]<|k||\sin\theta_{k}|+b,\quad\textrm{for }|k|>R_{1},\label{eq:CausalAsym}
\end{equation}
where $k=|k|e^{i\theta_{k}}$ and $C(\hat{\mathrm{n}})=|C(\hat{\mathrm{n}})|e^{i\theta_{c}}$.
Note that the angle $\theta_{k}$ is arbitrary and the parameter $\epsilon>0$
can be arbitrarily small. If $\alpha>1$, we can choose an appropriate
$\theta_{k}$ such that $|C(\hat{\mathrm{n}})|\sin(\alpha\theta_{k}+\theta_{c})-\epsilon>0$
and 
\begin{equation}
|k|^{\alpha}[|C(\hat{\mathrm{n}})|\sin(\alpha\theta_{k}+\theta_{c})-\epsilon]\gg|k||\sin\theta_{k}|+b,\quad\textrm{as }k\rightarrow\infty,
\end{equation}
which contradicts with inequality (\ref{eq:CausalAsym}). Hence $\alpha\leq1$
is a necessary condition.

Second, we prove that $\alpha\notin(0,1)$. Substituting $\theta_{k}=l\pi$
with $l=0,\pm1,\pm2,...$ into (\ref{eq:CausalAsym}), we find 
\begin{equation}
|k|^{\alpha}[|C(\hat{\mathrm{n}})|\sin(\alpha l\pi+\theta_{c})-\epsilon]<b,\quad\textrm{for }|k|>R_{1}.
\end{equation}
For $\alpha>0$, the above inequality leads to 
\begin{equation}
\sin(\alpha l\pi+\theta_{c})\leq0.\label{eq:CausalSineIneq}
\end{equation}
When $l=0$, we obtain $\mathrm{Im}\ C(\hat{\mathrm{n}})=|C(\hat{\mathrm{n}})|\sin\theta_{c}\leq0$.
Because the sine function is periodic, it is enough to consider the
range $\theta_{c}\in[-\pi,0]$ in the inequality (\ref{eq:CausalSineIneq}).
If $\alpha\in(0,1)$, we can choose 
\begin{equation}
l=\left\lfloor \frac{-\theta_{c}}{\pi\alpha}\right\rfloor +1,
\end{equation}
where $\left\lfloor x\right\rfloor $ is an integer obeying $x-1<\left\lfloor x\right\rfloor \leq x$
for $x\in\mathbb{R}$, such that $0<\alpha l\pi+\theta_{c}<\pi$ violating
the inequality (\ref{eq:CausalSineIneq}). Thus we have $\alpha\notin(0,1).$

Third, we show that $C(\hat{\mathrm{n}})\in[-1,0)\cup(0,1]$ if $\alpha=1$.
Let $\alpha=1$ in (\ref{eq:CausalAsym}). We find 
\begin{equation}
|C(\hat{\mathrm{n}})|\sin(\theta_{k}+\theta_{c})<|\sin\theta_{k}|+\epsilon+\frac{b}{|k|},\quad\textrm{for }|k|>R_{1}.
\end{equation}
Thus for any small constant $\epsilon^{\prime}>0$, there exists $R_{2}\geq R_{1}$
such that 
\begin{equation}
|C(\hat{\mathrm{n}})|\sin(\theta_{k}+\theta_{c})<|\sin\theta_{k}|+\epsilon^{\prime},\quad\textrm{for }|k|>R_{2}.
\end{equation}
Letting $\theta_{k}=\pi/2-\theta_{c},\pi,0$, we have 
\begin{equation}
|C(\hat{\mathrm{n}})|<|\sin\theta_{k}|+\epsilon^{\prime}<1+\epsilon^{\prime},\label{eq:CausAsyproof-1}
\end{equation}
and 
\begin{equation}
\pm|C(\hat{\mathrm{n}})|\sin\theta_{c}<\epsilon^{\prime}.\label{eq:CausAsyproof-2}
\end{equation}
The inequalities (\ref{eq:CausAsyproof-1},\ref{eq:CausAsyproof-2})
require that $C(\hat{\mathrm{n}})\in[-1,0)\cup(0,1]$.

Fourth, we prove that if $\alpha=1$, then the next leading order
of $\omega(\vec{k})$ is at most $\mathcal{O}(|k|^{0})$. The general
form of dispersion relation is given by 
\begin{equation}
\omega=C(\hat{\mathrm{n}})k+D(\hat{\mathrm{n}})k^{\beta}+\mathcal{O}(|k|^{\beta^{\prime}}),
\end{equation}
where $\beta^{\prime}<\beta<1$, and $D(\hat{\mathrm{n}})\neq0$ is
independent of $k$. Let $\theta_{k}=l\pi$ with $l=0,\pm1,\pm2,...$.
Using the inequality (\ref{eq:SufCausality}), we find that for any
small $\epsilon>0$, there exists $R_{3}>0$ satisfying 
\begin{equation}
|k|^{\beta}[|D(\hat{\mathrm{n}})|\sin(\beta l\pi+\theta_{d})-\epsilon]<b,\quad\text{for }|k|>R_{3}.
\end{equation}
Assuming that $0<\beta<1$, the above inequality requires 
\begin{equation}
\sin(\beta l\pi+\theta_{d})\leq0.\label{eq:CausalIneq_2}
\end{equation}
We immediately find $\sin\theta_{d}\leq0$. Then (\ref{eq:CausalIneq_2})
has the same structure as (\ref{eq:CausalSineIneq}). The argument
in the second step also applies to (\ref{eq:CausalIneq_2}). This
means (\ref{eq:CausalIneq_2}) cannot hold for $0<\beta<1$, that
is a contradiction. Therefore $\beta\leq0$. \end{proof}

\section{Key inequality (\ref{eq:Estimate1}) for causality criterion \label{sec:ProofInequality1}}

We prove inequality (\ref{eq:Estimate1}) here. For an $n\times n$
matrix, we introduce the spectral norm of matrix $||\mathrm{M}||=\sup_{|V^{\prime}|=1}|\mathrm{M}V^{\prime}|$
where $V^{\prime}$ is an arbitrary vector with $|V^{\prime}|=1$
\citep{Datta2010:NLinear}.

If $n=1$, we get $\mathrm{M}(\vec{k})=i\omega$ and 
\begin{equation}
|e^{-\mathrm{M}(\vec{k})t}|=e^{(\mathrm{Im}\ \omega)t},\quad\textrm{for }n=1.\label{eq:Appn1}
\end{equation}

Next, we consider the case $n>1$. For a fixed complex $\vec{k}$,
we have the identity (see chapter 7.3 of Ref. \citep{Evans2010PartialDE})
\begin{equation}
e^{-\mathrm{M}(\vec{k})t}=\frac{1}{2\pi i}\int_{\mathcal{C}}dze^{-zt}[z-\mathrm{M}(\vec{k})]^{-1},
\end{equation}
where the path $\mathcal{C}$ is defined as follows. Draw circles
of radius $\epsilon\in(0,1)$ centered at $i\omega_{1},i\omega_{2},...,i\omega_{n}$,
which are the eigenvalues of the $n\times n$ matrix $\mathrm{M}(\vec{k})$.
Then $\mathcal{C}$ is defined as the boundary of the union set of
these circles, traversed anticlockwise. We find 
\begin{eqnarray}
|e^{-zt}| & \leq & \exp\left\{ \left[\max_{1\leq i\leq n}(\mathrm{Im}\ \omega_{i})+\epsilon\right]t\right\} ,\label{eq:AppDet}\\
|\det[z-\mathrm{M}(\vec{k})]| & = & \prod_{j=1}^{n}|z-i\omega_{j}|\geq\epsilon^{n},\quad\textrm{for }z\in\mathcal{C}.
\end{eqnarray}
The spectral norm of the inverse matrix can be evaluated through 
\begin{eqnarray}
||[z-\mathrm{M}(\vec{k})]^{-1}|| & = & \frac{||\mathrm{ad}[z-\mathrm{M}(\vec{k})]||}{|\det[z-\mathrm{M}(\vec{k})]|}\nonumber \\
 & \leq & \epsilon^{-n}||\mathrm{ad}[z-\mathrm{M}(\vec{k})]||,\quad\textrm{for }z\in\mathcal{C},\label{eq:AppInv}
\end{eqnarray}
where $\mathrm{ad}[z-\mathrm{M}(\vec{k})]$ stands for the adjugate
matrix of $z-\mathrm{M}(\vec{k})$. Note that the spectral norm is
not greater than Frobenius norm \citep{Datta2010:NLinear}, i.e.,
\begin{equation}
||\mathrm{ad}[z-\mathrm{M}(\vec{k})]||\leq\left(\sum_{i,j=1}^{n}|\{\mathrm{ad}[z-\mathrm{M}(\vec{k})]\}_{ij}|^{2}\right)^{1/2}.
\end{equation}
In the following we estimate the matrix element $|\{\mathrm{ad}[z-\mathrm{M}(\vec{k})]\}_{ij}|$,
which equals the absolute value of the determinant of the $(n-1)\times(n-1)$
matrix obtained by deleting the $j$-th row and $i$-th column from
$z-\mathrm{M}(\vec{k})$.

To proceed further, we now show that $z\in\mathcal{C}$ is bounded
by 
\begin{equation}
|z|\leq\mathrm{const}\cdot(1+|\vec{k}|^{N}).\label{eq:Appz}
\end{equation}
Since the elements in $\mathrm{M}(\vec{k})$ fulfill 
\begin{equation}
|\mathrm{M}_{ij}(\vec{k})|\leq\mathrm{const}\cdot(1+|\vec{k}|^{N}),\label{eq:inequality_temp}
\end{equation}
 each eigenvalue $i\omega_{j}$ of $\mathrm{M}(\vec{k})$ is bounded
by $|\omega_{j}|\leq\mathrm{const}\cdot(1+|\vec{k}|^{N})$.  Recalling
that $|z|\leq|\omega_{j}|+1$, we then obtain inequality (\ref{eq:Appz}).

The inequalities, (\ref{eq:Appz}) and (\ref{eq:inequality_temp}),
give us 
\begin{equation}
|\{\mathrm{ad}[z-\mathrm{M}(\vec{k})]\}_{ij}|\leq\mathrm{const}\cdot(1+|\vec{k}|^{N(n-1)}),\quad\textrm{for }z\in\mathcal{C}.
\end{equation}
Hence,
\begin{eqnarray}
||\mathrm{ad}[z-\mathrm{M}(\vec{k})]|| & \leq & \left(\sum_{i,j=1}^{n}|\{\mathrm{ad}[z-\mathrm{M}(\vec{k})]\}_{ij}|^{2}\right)^{1/2}\nonumber \\
 & \leq & \mathrm{const}\cdot(1+|\vec{k}|^{N(n-1)}),\quad\textrm{for }z\in\mathcal{C}.\label{eq:AppAdjNorm}
\end{eqnarray}
According to (\ref{eq:AppDet},\ref{eq:AppInv},\ref{eq:AppAdjNorm}),
we find
\begin{eqnarray}
||[z-\mathrm{M}(\vec{k})]^{-1}|| & \leq & \epsilon^{-n}||\mathrm{ad}[z-\mathrm{M}(\vec{k})]||\nonumber \\
 & \leq & \mathrm{const}\cdot\epsilon^{-n}(1+|\vec{k}|^{N(n-1)}),\quad\textrm{for }z\in\mathcal{C}.
\end{eqnarray}
 and, then, obtain 
\begin{eqnarray}
||e^{-\mathrm{M}(\vec{k})t}|| & = & \frac{1}{2\pi}||\int_{\mathcal{C}}dze^{-zt}[z-\mathrm{M}(\vec{k})]^{-1}||\nonumber \\
 & \leq & \mathrm{const}\cdot\frac{2\pi\epsilon n}{2\pi}\epsilon^{-n}(1+|\vec{k}|^{N(n-1)})\exp\left\{ \left[\max_{1\leq i\leq n}(\mathrm{Im}\ \omega_{i})+\epsilon\right]t\right\} \nonumber \\
 & \leq & \mathrm{const}\cdot\epsilon^{1-n}(1+|\vec{k}|^{N(n-1)})\exp\left\{ \left[\max_{1\leq i\leq n}(\mathrm{Im}\ \omega_{i})+\epsilon\right]t\right\} ,\quad\textrm{for }n>1.\label{eq:AppNormn2}
\end{eqnarray}
Combining (\ref{eq:Appn1}) and (\ref{eq:AppNormn2}), we derive the
inequality (\ref{eq:Estimate1}). We emphasize that the $\epsilon\in(0,1)$
is arbitrary, and the ``const'' in (\ref{eq:AppNormn2}) is independent
of $t$, $\epsilon$, $\omega$, and $\vec{k}$.

\section{Three examples for the applications of our new theorems}

\subsection{Example 1: Tachyon field equation is causal but unstable \label{subsec:TachyonCausal}}

Tachyon field equation is similar to the Klein-Gordon field equation
but with an imaginary mass $im_{0}$ \citep{Aharonov:1969vu}, i.e.,
\begin{equation}
(\partial_{t}^{2}-\nabla^{2}-m_{0}^{2})\varphi(x)=0.\label{eq:TachyonEq}
\end{equation}
This equation can be converted into the form (\ref{eq:StandardFormFirstOrder}),
so we can use the results shown in this work to study the stability
and causality of Tachyon field.

The dispersion relations of (\ref{eq:TachyonEq}) read
\begin{equation}
\omega=\pm(\vec{k}^{2}-m_{0}^{2})^{1/2}.\label{eq:DispersionTachyon}
\end{equation}
Obviously, the Tachyon field equation is unstable since one mode has
positive imaginary part as $|\vec{k}|\rightarrow0$.

We then discuss the causality. From (\ref{eq:DispersionTachyon})
we can get  
\begin{eqnarray}
(\mathrm{Re}\ \omega)^{2}-(\mathrm{Im}\ \omega)^{2} & = & (\mathrm{Re}\ \vec{k})^{2}-(\mathrm{Im}\ \vec{k})^{2}-m_{0}^{2},\\
(\mathrm{Re}\ \omega)(\mathrm{Im}\ \omega) & = & (\mathrm{Re}\ \vec{k})\cdot(\mathrm{Im}\ \vec{k}),
\end{eqnarray}
which leads to 
\begin{eqnarray}
(\mathrm{Im}\ \omega)^{2} & \leq & (\mathrm{Im}\ \vec{k})^{2}-\frac{1}{2}(|\vec{k}|^{2}-m_{0}^{2})+\frac{1}{2}\sqrt{(|\vec{k}|^{2}+m_{0}^{2})^{2}-4m_{0}^{2}(\mathrm{Re}\ \vec{k})^{2}}\nonumber \\
 & \leq & (\mathrm{Im}\ \vec{k})^{2}+m_{0}^{2}.
\end{eqnarray}
Thus $\mathrm{Im}\ \omega\leq|\mathrm{Im}\ \vec{k}|+m_{0}$ satisfies
the improved causality criterion (\ref{eq:SufCausality}) and its
equivalent version (\ref{eq:OmegaBound}), indicating that the Tachyon
field equation (\ref{eq:TachyonEq}) is causal indeed.

In fact, there is a simpler way to show that the Tachyon field equation
is causal. Since (\ref{eq:TachyonEq}) is isotropic, the dispersion
relations are independent of the direction of $\vec{k}$. It is enough
to study the $1+1$ dimensional propagation. Letting $\vec{k}=(k,0,0)$
in (\ref{eq:DispersionTachyon}), the asymptotic behavior at $k\rightarrow\infty$
is given by 
\begin{equation}
\omega=\pm k+\mathcal{O}(|k|^{-1}),
\end{equation}
which satisfies the asymptotic behaviors in Theorem \ref{TheoremAsymp}.
For any $\epsilon>0$, we can find a large number $R$ such that $\mathrm{Im}[\mathcal{O}(|k|^{-1})]<\epsilon$
for $|k|>R$. Therefore $\mathrm{Im}\ \omega<|\mathrm{Im}\ k|+\epsilon$
for $|k|>R$, and then the Tachyon field equation is causal according
to the causality criterion (\ref{eq:SufCausality}).

\subsection{Example 2: Stability and causality of simplified MIS theory with
bulk viscosity \label{subsec:SCMIS}}

Let us consider the simplified MIS theory with bulk viscous pressure
only. The energy-momentum tensor is given by \citep{Israel:1979}
\begin{equation}
\Theta^{\mu\nu}=eu^{\mu}u^{\nu}-(p+\Pi)\Delta^{\mu\nu},\quad\partial_{\mu}\Theta^{\mu\nu}=0,\label{eq:EMT_bulk}
\end{equation}
where the constitutive equation for bulk viscosity $\Pi$ is 
\begin{equation}
\tau_{\Pi}u^{\mu}\partial_{\mu}\Pi+\Pi=-\zeta\partial_{\mu}u^{\mu}.\label{eq:MISbulk}
\end{equation}
The variables $e$, $p$, $u^{\mu}$, $\tau_{\Pi}\neq0$, and $\zeta\neq0$
represent energy density, pressure, velocity of a fluid cell, relaxation
time for bulk viscous pressure, and bulk viscosity coefficient, respectively.
The projection tensor $\Delta^{\mu\nu}$ is defined as $\Delta^{\mu\nu}\equiv\eta^{\mu\nu}-u^{\mu}u^{\nu}$
with $\eta^{\mu\nu}=\mathrm{diag}\{+,-,-,-\}$. In the following,
we employ the linear mode analysis and characteristic method to investigate
the stability and causality of the simplified MIS theory. 

\subsubsection{Linear stability and causality conditions via linear mode analysis}

We now analyze the linear stability and causality of the simplified
MIS theory through linear mode analysis. We analyze the conventional
stability criterion (\ref{eq:ineq_orginal_01}) combing the Theorem
\ref{TheoremAsymp} in rest frame and extend it across all IFR by
applying Theorem \ref{TheoremSCinMovingF}. 

Consider small perturbations atop an irrotational equilibrium state,
e.g., $\delta e$, $\delta\Pi$, and $\delta u^{\mu}$. The evolution
of these perturbations is governed by the conservation equation $\partial_{\mu}\Theta^{\mu\nu}=0$
and constitutive equation (\ref{eq:MISbulk}). To linear order, we
obtain  
\begin{eqnarray}
u^{\mu}\partial_{\mu}\delta e+(e+p)\partial_{\mu}\delta u^{\mu} & = & 0,\label{eq:LinearMEq1}\\
(e+p)u^{\mu}\partial_{\mu}\delta u^{\alpha}-c_{s}^{2}\Delta^{\mu\alpha}\partial_{\mu}\delta e-\Delta^{\mu\alpha}\partial_{\mu}\delta\Pi & = & 0,\label{eq:LinearMEq2}\\
\tau_{\Pi}u^{\mu}\partial_{\mu}\delta\Pi+\delta\Pi+\zeta\partial_{\mu}\delta u^{\mu} & = & 0,\label{eq:LinearMEq3}
\end{eqnarray}
where $c_{s}^{2}\equiv\frac{dp}{de}\neq0$ is the speed of sound.
In the rest frame where $u^{\mu}=(1,\vec{0})$, the equations (\ref{eq:LinearMEq1}-\ref{eq:LinearMEq3})
become isotropic. In such case, we only need to consider 1+1 dimensional
propagation. Assuming the perturbations are proportional to $e^{-i\omega t+ikx}$,
the nonzero dispersion relations are the solutions to  
\begin{equation}
i\tau_{\Pi}\omega^{3}-\omega^{2}-i(\overline{\zeta}+c_{s}^{2}\tau_{\Pi})\omega k^{2}+c_{s}^{2}k^{2}=0,\label{eq:DisEqMIS}
\end{equation}
with $\overline{\zeta}\equiv\zeta/(e+p)$. 

Based on Eq. (\ref{eq:DisEqMIS}), it is straightforward to derive
the stability and causality conditions in the rest frame. For stability,
the Routh-Hurwitz criterion \citep{Kovtun:2019hdm} directly yields
the following conditions  
\begin{equation}
\tau_{\Pi}>0,\ \zeta>0,\ c_{s}^{2}>0.\label{eq:LinearSCMIS}
\end{equation}
For causality, we note that the solutions to (\ref{eq:DisEqMIS})
exhibit the following asymptotic behaviors as $k\rightarrow\infty$,
\begin{eqnarray}
\omega & = & \pm k\sqrt{c_{s}^{2}+\frac{\overline{\zeta}}{\tau_{\Pi}}}-\frac{i\overline{\zeta}}{2\tau_{\Pi}(\overline{\zeta}+c_{s}^{2}\tau_{\Pi})}+\mathcal{O}(k^{-1}),\label{eq:AsymMIS1}\\
\omega & = & -\frac{ic_{s}^{2}}{\overline{\zeta}+c_{s}^{2}\tau_{\Pi}}+\mathcal{O}(k^{-1}),\label{eq:AsymMIS2}
\end{eqnarray}
Then we can straightforwardly obtain the causality condition \cite{Pu:2009fj}
\begin{equation}
0\leq c_{s}^{2}+\frac{\zeta}{\tau_{\Pi}(e+p)}\leq1,\label{eq:LinearCCMIS}
\end{equation}
from (\ref{eq:AsymMIS1}-\ref{eq:AsymMIS2}) by applying the criterion
(\ref{eq:SufCausality}) in Theorem \ref{TheoremSuffCausa} (also
see Eqs. (\ref{eq:Iso_1}-\ref{eq:Iso_3}) in Theorem \ref{TheoremAsymp}
and Sec. \ref{subsec:Relation-between-the}).

When considering the stability and causality in other IFR, there are
no additional conditions required beyond conditions (\ref{eq:LinearSCMIS})
and (\ref{eq:LinearCCMIS}). For causality, Theorem \ref{TheoremCausOtherF}
asserts that (\ref{eq:LinearCCMIS}) is sufficient to guarantee causality
in any IFR. For stability, however, we emphasize that condition (\ref{eq:LinearSCMIS})
alone can not ensure stability across all IFR. Instead, the intersection
of conditions (\ref{eq:LinearSCMIS}) and (\ref{eq:LinearCCMIS})
establish the conditions for covariant stability, as pointed out by
Theorem \ref{TheoremSCinMovingF}.

\subsubsection{Nonlinear causality condition via characteristic method}

The above analysis focus on the linear regime. In contrast, we now
discuss the nonlinear causality using the characteristic method \citep{Floerchinger:2017cii,Bemfica:2020zjp}.

Without linear order approximation, the full hydrodynamic equations
(\ref{eq:EMT_bulk}, \ref{eq:MISbulk}) can be put in the matrix form
 
\begin{equation}
\mathcal{M}(\partial)\Phi=\mathcal{N}(\Phi),
\end{equation}
with $\Phi=(e,u^{\nu},\Pi)^{\mathrm{T}}$, $\mathcal{N}(\Phi)=(0,0,-\Pi)^{\mathrm{T}}$,
and 
\begin{equation}
\mathcal{M}(\partial)=\left[\begin{array}{ccc}
u^{\mu}\partial_{\mu} & (e+p+\Pi)\partial_{\nu} & 0\\
-c_{s}^{2}\Delta^{\alpha\mu}\partial_{\mu} & (e+p+\Pi)\delta_{\nu}^{\alpha}u^{\mu}\partial_{\mu} & -\Delta^{\alpha\mu}\partial_{\mu}\\
0 & \zeta\partial_{\nu} & \tau_{\Pi}u^{\mu}\partial_{\mu}
\end{array}\right].
\end{equation}
Causality demands that the roots of $\det[\mathcal{M}(\psi)]=0$,
i.e., $\psi^{\mu}=(\psi^{0}(\psi^{i}),\psi^{i})$, satisfies the constraints:
(i) $\psi^{\mu}$ is real and (ii) $\psi^{\mu}\psi_{\mu}\leq0$ \citep{Floerchinger:2017cii}.

We will derive the nonlinear causality conditions based on the constraints
mentioned above. The characteristic determinant $\det[\mathcal{M}(\psi)]$
is given by  
\begin{eqnarray}
\det[\mathcal{M}(\psi)] & = & \det\left[\begin{array}{ccc}
u^{\mu}\psi_{\mu} & (e+p+\Pi)\psi_{\nu} & 0\\
-c_{s}^{2}\Delta^{\alpha\mu}\psi_{\mu} & (e+p+\Pi)\delta_{\nu}^{\alpha}u^{\mu}\psi_{\mu} & -\Delta^{\alpha\mu}\psi_{\mu}\\
0 & \zeta\psi_{\nu} & \tau_{\Pi}u^{\mu}\psi_{\mu}
\end{array}\right]\nonumber \\
 & = & (e+p+\Pi)^{3}B^{4}\left\{ \tau_{\Pi}(e+p+\Pi)B^{2}+[\tau_{\Pi}c_{s}^{2}(e+p+\Pi)+\zeta]V^{2}\right\} ,
\end{eqnarray}
with $B\equiv u_{\nu}\psi^{\nu}$ and $V^{\mu}\equiv\Delta^{\mu\nu}\psi_{\nu}$.
Then the equation $\det[\mathcal{M}(\psi)]=0$ leads to 
\begin{equation}
B=0,\quad B^{2}=-V^{2}\left[c_{s}^{2}+\frac{\zeta}{\tau_{\Pi}(e+p+\Pi)}\right].
\end{equation}
Note that $V^{2}\leq0$ since $V^{\mu}$ is not time-like. Thus $\psi^{\mu}=(\psi^{0}(\psi^{i}),\psi^{i})$
is real for any real $\psi^{i}$ if 
\begin{equation}
c_{s}^{2}+\frac{\zeta}{\tau_{\Pi}(e+p+\Pi)}\geq0.
\end{equation}
For the constraint $\psi^{\mu}\psi_{\mu}\leq0$, we find 
\begin{eqnarray}
\psi^{\mu}\psi_{\mu} & = & B^{2}+V^{2}\nonumber \\
 & = & V^{2}\left[1-c_{s}^{2}-\frac{\zeta}{\tau_{\Pi}(e+p+\Pi)}\right]\leq0,
\end{eqnarray}
which gives,
\begin{equation}
c_{s}^{2}+\frac{\zeta}{\tau_{\Pi}(e+p+\Pi)}\leq1.
\end{equation}

In summary, the nonlinear causality condition is given by \citep{Floerchinger:2017cii}
\begin{equation}
0\leq c_{s}^{2}+\frac{\zeta}{\tau_{\Pi}(e+p+\Pi)}\leq1.\label{eq:NonlinearCMIS}
\end{equation}
As expected, when the system reaches an equilibrium state, i.e., $\Pi\rightarrow0$,
the nonlinear causality condition (\ref{eq:NonlinearCMIS}) reduces
to the linear one (\ref{eq:LinearCCMIS}).

\subsection{Example 3: General expression for causal and stable asymptotic dispersion
relations in isotropic systems \label{sec:ProofTheoremAsym}}

Here we discuss the general expression for causal and stable asymptotic
dispersion relations in isotropic systems. The results have been shown
in Theorem 5. By dominant balance \citep{bender1999advanced}, one
can find that the asymptotic solutions to (\ref{eq:PolynomialP})
with $\vec{k}=k\hat{\mathrm{n}}$ are given by 
\begin{equation}
\omega=Ck^{\alpha}+\mathcal{O}(|k|^{\beta}),\ \textrm{as }k\rightarrow\infty,\label{eq:AsympSol}
\end{equation}
where $\alpha,C$, and $\beta<\alpha$ are independent of $k$ and
$\hat{\mathrm{n}}$ in isotropic cases.  We assume $C\neq0$ from
now on.

First, we prove that 
\begin{equation}
\mathrm{Im}\ \omega\leq|\mathrm{Im}\ k|,\label{eq:ProofIsoBound}
\end{equation}
holds for any $\hat{\mathrm{n}}$ in isotropic cases. Let $k=k_{R}+ik_{I}$
and $\hat{\mathrm{n}}=\hat{\mathrm{n}}_{R}+i\hat{\mathrm{n}}_{I}$,
where $k_{R},k_{I},\hat{\mathrm{n}}_{R},\hat{\mathrm{n}}_{I}$ are
real, and 
\begin{equation}
\hat{\mathrm{n}}_{R}\cdot\hat{\mathrm{n}}_{I}=0,\ |\hat{\mathrm{n}}_{R}|^{2}-|\hat{\mathrm{n}}_{I}|^{2}=1.
\end{equation}
Then 
\begin{eqnarray}
|\mathrm{Im}\ \vec{k}| & = & |\mathrm{Im}\ k\hat{\mathrm{n}}|\nonumber \\
 & = & |k_{I}\hat{\mathrm{n}}_{R}+k_{R}\hat{\mathrm{n}}_{I}|\nonumber \\
 & = & \sqrt{k_{I}^{2}+(k_{R}^{2}+k_{I}^{2})|\hat{\mathrm{n}}_{I}|^{2}}\nonumber \\
 & \geq & |k_{I}|=|\mathrm{Im}\ k|,
\end{eqnarray}
where the equal sign holds for $\hat{\mathrm{n}}_{I}=0$. Because
$\omega$ is independent of $\hat{\mathrm{n}}$ in isotropic cases,
the inequality $\mathrm{Im}\ \omega\leq|\mathrm{Im}\ \vec{k}|$ leads
to (\ref{eq:ProofIsoBound}).

Second, we prove $\alpha\leq1$ and $\mathrm{Im}\ C\leq0$. For any
small number $\epsilon>0$, there exists a large enough number $R>0$
such that $\mathrm{Im}[\mathcal{O}(|k|^{\beta})]>-\epsilon|k|^{\alpha}$
for $\beta<\alpha$ and $|k|>R$. Then, the condition (\ref{eq:ProofIsoBound})
gives 
\begin{equation}
|k|^{\alpha}\left[|C|\sin(\alpha\theta_{k}+\theta_{c})-\epsilon\right]<|k||\sin\theta_{k}|,\label{eq:TheoremAsymp-1}
\end{equation}
for $|k|>R$, where $C=|C|e^{i\theta_{c}}$ and $k=|k|e^{i\theta_{k}}$.
Because $\theta_{k}$ is an arbitrary real number and $\epsilon$
can be arbitrarily small, we have $\alpha\leq1$ and $\mathrm{Im}\ C=|C|\sin\theta_{c}\leq0$.

Third, we prove that $\alpha$ is an integer. By letting $\theta_{k}=l\pi$
with $l=0,\pm1,\pm2,...$, the inequality (\ref{eq:TheoremAsymp-1})
becomes 
\begin{equation}
|C|\sin(\alpha l\pi+\theta_{c})-\epsilon<0,\label{eq:TheoremAsymp-2}
\end{equation}
for $|k|>R$. Since $\epsilon>0$ can be arbitrarily small, (\ref{eq:TheoremAsymp-2})
leads to 
\begin{equation}
\sin(\alpha l\pi+\theta_{c})\leq0,\label{eq:TheoremAsymp-3}
\end{equation}
which must hold for any integer $l$. (\ref{eq:TheoremAsymp-3}) requires
that $\alpha$ must be an integer. To show this, we suppose, by contradiction,
that $\alpha$ is not an integer. Because of the periodicity of the
sine function in (\ref{eq:TheoremAsymp-3}), we only need to consider
the following range,
\begin{equation}
\alpha\in(-1,0)\cup(0,1),\quad\theta_{c}\in[-\pi,0].
\end{equation}
Let
\begin{equation}
l=\pm\left\lfloor \frac{-\theta_{c}}{\pi|\alpha|}\right\rfloor \pm1,\label{eq:LValue}
\end{equation}
where $\left\lfloor x\right\rfloor $ is an integer obeying $x-1<\left\lfloor x\right\rfloor \leq x$
for $x\in\mathbb{R}$, and the sign ``$+$'' and ``$-$'' corresponds
to $\alpha\in(0,1)$ and $\alpha\in(-1,0)$, respectively.  With
(\ref{eq:LValue}), we find $0<\alpha l\pi+\theta_{c}<\pi$ and $\sin(\alpha l\pi+\theta_{c})>0$,
contradicting (\ref{eq:TheoremAsymp-3}). Thus, $\alpha$ is an integer.

Fourth, we derive the constraints for $C$ in different cases. If
$\alpha=0,-2,-4,-6,...$, the inequality (\ref{eq:TheoremAsymp-3})
implies $\mathrm{Im}\ C=|C|\sin\theta_{c}\leq0$. If $\alpha=-1,-3,-5,-7,...$,
the inequality (\ref{eq:TheoremAsymp-3}) implies $\mathrm{Im}\ C=|C|\sin\theta_{c}=0$.
The special case is $\alpha=1$. Substituting $\alpha=1$ into (\ref{eq:TheoremAsymp-1}),
we obtain 
\begin{equation}
|C|\sin(\theta_{k}+\theta_{c})-\epsilon<|\sin\theta_{k}|,
\end{equation}
for $|k|>R$. When $\theta_{k}=\pi/2-\theta_{c},\pi,0$, we have 
\begin{equation}
|C|<|\sin(\pi/2-\theta_{c})|+\epsilon\leq1+\epsilon,
\end{equation}
and 
\begin{equation}
\pm|C|\sin\theta_{c}\leq\epsilon.
\end{equation}
Since $\epsilon$ can be arbitrarily small and $C$ is nonzero, we
obtain $C\in[-1,0)\cup(0,1]$ for $\alpha=1$.

The last step is to derive the next leading order of (\ref{eq:AsympSol})
when $\alpha=1,-1,-3,...$. Assume that the next leading order is
$Dk^{\beta}$ with a nonzero constant $D$, i.e., 
\begin{equation}
\omega=Ck^{\alpha}+Dk^{\beta}+\mathcal{O}(|k|^{\beta^{\prime}}),\label{eq:NextLO}
\end{equation}
where $\mathrm{Im}\ C=0$ and $\beta^{\prime}<\beta<\alpha$. Let
$k=|k|e^{il\pi}$ with $l=0,\pm1,\pm2,...$, and $D=|D|e^{i\theta_{d}}$.
Imposing the condition $\mathrm{Im}\ \omega\leq|\mathrm{Im}\ k|$
on (\ref{eq:NextLO}), we get 
\begin{equation}
\sin(\beta l\pi+\theta_{d})\leq0,
\end{equation}
which has the same structure as (\ref{eq:TheoremAsymp-3}). The case
of $l=0$ gives $\mathrm{Im}\ D=|D|\sin\theta_{d}\leq0$. Then using
the same argument on the third step, one can show that $\beta$ must
be an integer, i.e., $\beta\in\{\alpha-1,\alpha-2,...\}$. 

Eventually, we list the general expression for causal and stable asymptotic
dispersion relations
\begin{eqnarray}
\omega & = & c_{1}k+d_{1}+\mathcal{O}(|k|^{a}),\nonumber \\
\omega & = & c_{2}k^{-2m-1}+d_{2}k^{-2m-2}+\mathcal{O}(|k|^{a-2m-2}),\nonumber \\
\omega & = & c_{3}k^{-2m}+\mathcal{O}(|k|^{a-2m}),
\end{eqnarray}
where $a<0$, $m=0,1,2,...$, and $c_{i},d_{i}$ are constants obeying
$c_{1}\in[-1,1]$, $\mathrm{Im}\ c_{2}=0$, $\mathrm{Im}\ c_{3}\leq0$,
$\mathrm{Im}\ d_{1,2}\leq0$.
\end{document}